\documentclass[a4paper,UKenglish,cleveref, autoref]{lipics-v2019}
\usepackage{comment}
\usepackage[utf8]{inputenc}
\usepackage{amsfonts}
\usepackage{amsthm,amsmath,amssymb,mathtools}
\usepackage{algorithm}
\usepackage[noend]{algpseudocode}
\usepackage{pifont}
\usepackage{qtree}
\usepackage[english,rounding]{rccol}
\usepackage{rotating}
\usepackage{tabularx}

\hideLIPIcs

\usepackage{algpseudocode}
\usepackage{enumitem}
\usepackage{graphicx}
\usepackage{comment}
\usepackage{xcolor}
\usepackage{wrapfig}

\newcommand{\etal}{{et al.}}

 \newcommand{\joachim}[1]{\textcolor{purple}{JG: #1}}
 
\newcommand{\R}{\mathbb{R}}

\def\polylog{\operatorname{polylog}}
\DeclareMathOperator*{\argmax}{arg\,max}

\theoremstyle{definition}

\title{Approximating the packedness of polygonal curves} 
\titlerunning{Packedness of curves}

\Copyright{-}

\author{Joachim Gudmundsson}{The University of Sydney, Australia}
{}{}{}
\author{Yuan Sha}{The University of Sydney, Australia}
{}{}{}
\author{Sampson Wong}{The University of Sydney, Australia}
{}{}{} 
\authorrunning{Gudmundsson et al.} 

\Copyright{Joachim Gudmundsson, Yuan Sha and Sampson Wong}



\EventEditors{}
\EventNoEds{0}
\EventLongTitle{}
\EventShortTitle{}
\EventAcronym{}
\EventYear{2020}
\EventDate{}
\EventLocation{}
\EventLogo{}
\SeriesVolume{}
\ArticleNo{}

\begin{document}
\maketitle

\begin{abstract}
    In 2012 Driemel \etal~\cite{DBLP:journals/dcg/DriemelHW12} introduced the concept of $c$-packed curves as a realistic input model. In the case when $c$ is a constant they gave a near linear time $(1+\varepsilon)$-approximation algorithm for computing the Fr\'echet distance between two $c$-packed polygonal curves. Since then a number of papers have used the model. 

    In this paper we consider the problem of computing the smallest $c$ for which a given polygonal curve in $\R^d$ is $c$-packed. We present two approximation algorithms. The first algorithm is a $2$-approximation algorithm and runs in $O(dn^2 \log n)$ time. In the case $d=2$ we develop a faster algorithm that returns a $(6+\varepsilon)$-approximation and runs in $O((n/\varepsilon^3)^{4/3} \polylog (n/\varepsilon)))$ time.
    
    We also implemented the first algorithm and computed the approximate packedness-value for 16 sets of real-world trajectories. The experiments indicate that the notion of $c$-packedness is a useful realistic input model for many curves and trajectories.  
    \end{abstract}

\section{Introduction}
Worst-case analysis often fails to accurately estimate the performance of an algorithm for real-world data. One reason for this is that the traditional analysis of algorithms and data structures is only done in terms of the number of elementary objects in the input; it does not take into account their distribution. Problems with traditional analysis have led researchers to analyse algorithms under certain assumptions on the input~\cite{bsvk-rimga-02}, which are often satisfied in practice. By doing this, complicated hypothetical inputs are hopefully precluded, and the worst-case analysis yields bounds which better reflect the behaviour of the algorithms in practical situations. 

In computational geometry, realistic input models were introduced by van der Stappen and Overmars~\cite{so-mpafo-94} in 1994. They studied motion planning among fat obstacles. Since then a range of models have been proposed, including uncluttered scenes~\cite{b-lsbsp-00}, low density~\cite{sobv-mpelo-98}, simple-cover complexity~\cite{mms-qsrs-97}, to name a few. De Berg et al.~\cite{bsvk-rimga-02} gave 
algorithms for computing the model parameters for planar polygonal scenes. In their paper they motivated why such algorithms are important. 
\begin{itemize}
\item To verify whether a certain model is appropriate for a certain application domain.
\item Some algorithms require the value of the model parameter as input in order to work correctly, e.g. the range searching data structure for fat objects developed by Overmars and van der Stappen~\cite{os-rsplaf-96}.
\item Computing the model parameters of a given input can be useful for selecting the algorithm best tailored to that specific input.
\end{itemize}


In this paper we will study polygonal curves in $\R^d$. The Fr\'{e}chet distance~\cite{frechet} is probably the most popular distance measure for curves. In 1995, Alt and Godau~\cite{F} presented an $O(n^2 \log n)$ time algorithm for computing the Fr\'{e}chet distance between two polygonal curves of complexity $n$. This was later improved by Buchin et al.~\cite{bbmw-fswd-17} who showed that the continuous Fr\'{e}chet distance can be computed in $O(n^2 \sqrt{\log n}(\log \log n)^{3/2})$ expected time. Any attempt to find a much faster algorithm was proven to be futile when Bringmann~\cite{b-wwdtt-14} showed that, assuming the Strong Exponential Time Hypothesis, the Fr\'{e}chet distance cannot be computed in strongly subquadratic time, i.e., in time $O(n^{2-\varepsilon})$ for any $\varepsilon > 0$. 

In an attempt to break the quadratic lower bound for realistic curves, Driemel \etal~\cite{DBLP:journals/dcg/DriemelHW12} introduced a new family of realistic curves, so-called $c$-packed curves, which since then has gained considerable attention~\cite{cdgnw-ammrf-11,dh-jydp-13,dk-pefd-18,gs-faafm-15,hr-fdre-14}. A curve $\pi$ is $c$-packed if for any ball $B$, the length of the portion of $\pi$ contained in $B$ is at most $c$ times the radius of $B$.
In their paper they considered the problem of computing the Fr\'{e}chet distance between two $c$-packed curves and presented a $(1+\varepsilon)$-approximation algorithm with running time $O(\frac{cn}{\varepsilon} + cn \log n)$, which was later improved to $O(\frac{cn}{\sqrt{\varepsilon}} \log^2 (1/\varepsilon) + cn \log n)$ by Bringmann and K\"{u}nnemann~\cite{bk-iafdc-17}. 

Other models for realistic curves have also been studied. Closely related to $c$-packedness is $\gamma$-density which was introduced by Van der Stappen et al.~\cite{sobv-mpelo-98} for obstacles, and modified to polygonal curves in~\cite{DBLP:journals/dcg/DriemelHW12}.
A set of objects is $\gamma$-low-density, if for any ball of any radius, the number of objects intersecting the ball that are larger than the ball is less than $\gamma$. Aronov et al.~\cite{ahkww-fdcr-06} studied so-called backbone curves, which are used to model protein backbones in molecular biology. Backbone curves are required to have, roughly, unit edge length and a given minimal distance between any pair of vertices. Alt et al.~\cite{ akw-cdmp-04} introduced $\kappa$ straight curves, which are curves where the arc length between any two points on the curve is at most a constant $\kappa$ times their Euclidean distance. They also introduced $\kappa$-bounded curves which is a generalization of $\kappa$-straight curves. It has been shown~\cite{akklmss-cdsrp-08} that one can decide in $O(n \log n)$ time whether a given curve is backbone, $\kappa$-straight or $\kappa$-bounded. 

From the above discussion and the fact that the $c$-packed model has gained in popularity, we study two natural and important questions in this paper.
\begin{enumerate}
    \item Given a curve $\pi$, how fast can one (approximately) decide the smallest $c$ for which $\pi$ is $c$-packed?  
    \item Are real-world trajectory data $c$-packed for some reasonable value of $c$? 
\end{enumerate}

Vigneron~\cite{DBLP:journals/talg/Vigneron14} gave an FPTAS for optimizing the sum of algebraic functions. The algorithm can be applied to compute a $(1+\varepsilon)$ approximation of the $c$-packedness value of a polygonal curve in $\R^d$ in $O((\frac{n}{\varepsilon})^{d+2}\log^{d+2}\frac{n}{\varepsilon})$ time.

However, working with balls is complicated (see Section~\ref{ssec:preliminaries}) and in this paper we will therefore consider a simplified version of $c$-packedness. Instead of balls we will use ($d$-)cubes, that is, we say that a curve $\pi$ is $c$-packed if for any cube $S$, the length of the portion of $\pi$ contained in $S$ is at most $c \cdot r$, where $r$ is half the side length of $S$. Note that under this definition, a $c$-packed curve using the ``ball'' definition is a $(\sqrt{d}c)$-packed curve in the ``cube''' definition, while a $c$-packed curve using the ``cube'' definition is also a $c$-packed curve in the ``ball'' definition. From now on we will use the ``cube'' definition of $c$-packed curves. 

To the best of our knowledge the only known algorithm for computing packedness of a polygonal curve, apart from applying the tool by Vigneron~\cite{DBLP:journals/talg/Vigneron14}, is by Gudmundsson~\etal~\cite{gks-ahct-13} who gave a cubic time algorithm for polygonal curves in $\R^2$. They consider the problem of computing ``hotspots'' for a given polygonal curve, but their algorithm can also compute the packedness of a polygonal curve. We provide two sub-cubic time approximation algorithms for the packedness of a polygonal curve.

Our first result is a simple $O(dn^2 \log n)$ time $2$-approximation algorithm for $d$-dimensional polygonal curves. We also implemented this algorithm and tested it on 16 data sets to estimate the packedness value for real-world trajectory data. As expected the value varies wildly both between different data sets but also within the same data set. However, about half the data sets had an average packedness value less than 10, which indicates that $c$-packedness is a useful and realistic model for many real-world data sets.   

Our second result is a faster $O^*(n^{4/3})$ time\footnote{The $O^*$-notation omits $\polylog$ and $1/\varepsilon$ factors. } $(6+\varepsilon)$-approximation algorithm for polygonal curves in the plane. We achieve this faster algorithm by applying Callahan and Kosaraju's Well-Separated Pair Decomposition (WSPD) to select $O(n)$ squares, and then approximating the packedness values of these squares with a multi-level data structure. Note that our approach of building a data structure and then performing a linear number of square packedness queries solves a generalised instance of Hopcroft's problem. Hopcroft's problem asks: Given a set of $n$ points and $n$ lines in the plane, does any point lie on a line? An $\Omega(n^{4/3})$ lower bound for Hopcroft's problem was given by Erickson~\cite{DBLP:conf/cccg/Erickson95}. Hence, it is unlikely that our approach, or a similar approach, can lead to a considerably faster algorithm. 

\subsection{Preliminaries and our results} \label{ssec:preliminaries}
Let $\pi=\langle p_1, \ldots , p_n \rangle$ be a polygonal curve in $\R^d$ and let $s_i=(p_i,p_{i+1})$ for $1\leqslant i< n$. Let $H$ be a closed convex region in $\R^d$. The function $\Upsilon(H) = \sum_{i=1}^{n-1} |s_i \cap H|$ describes the total length of the trajectory $\pi$ inside $H$. In the original definition of $c$-packedness $H$ is a ball. 

As mentioned in the introduction, we will consider $H$ to be an axis-aligned cube instead of a ball. The reason for our choice was argued for $\R^2$ in~\cite{gks-ahct-13}, and for completeness we include their arguments here.  

If $H$ is a square, then each piece of $\Upsilon(H)$ is a simple linear function, i.e. is of the form $\gamma(x) = ax + b$ for some $a,b \in \R$. The description of each piece of $\Upsilon$ is constant size and can be evaluated in constant time. However, if $H$ is a disc, the intersection points of the boundary of $H$ with the trajectory $\pi$ are no longer simple linear equations in terms of the center and radius of $H$, so that $\Upsilon$ becomes a piecewise the sum of square roots of polynomial functions. These square root functions provide algebraic issues that cannot be easily resolved for maximising the function $\Upsilon(H)/r$. For this reason, we will consider $H$ to be a square instead of a disc.

The function $\Upsilon(H) = \sum_{i=1}^{n-1} |s_i \cap H|$ describes the total length of the polygonal curve inside $H$. 
Similarly, $\Psi(H) = \Upsilon(H)/r$ denotes the \emph{packedness value} of $H$. Our aim is to find a cube $H^*$ with centre at $p^*$ and radius $r^*$ that has the maximum packedness value for a given polygonal curve $\pi$. The radius of a cube is half the side length of the cube.

The following two theorems summarise the main results of this paper.
\begin{theorem}\label{theorem:simpleAlgorithmThm}
 Given a polygonal curve $\pi$ of size $n$ in $\R^d$, one can compute a $2$-approximate packedness value for $\pi$ in $O(dn^2\log n)$ time.
\end{theorem}

\begin{theorem}\label{theorem:complexAlgorithmThm}
Given a polygonal curve $\pi$ of size $n$ in $\R^2$ and a constant $\varepsilon$, with  $0<\varepsilon\leqslant1$, one can compute a $(6+\varepsilon)$-approximate packedness value for $\pi$ in 
$O((n/\varepsilon^3)^{4/3}\polylog(n/\varepsilon))$
time.
\end{theorem}


Theorem~\ref{theorem:simpleAlgorithmThm} is presented in Section~\ref{sec:2-approximation} and 
Theorem~\ref{theorem:complexAlgorithmThm} is presented in Section~\ref{sec:6-approximation}.
Experimental results on the packedness values for real world data sets are given in Section~\ref{sec:experiments}.


\section{A \texorpdfstring{$2$}{2}-approximation algorithm} \label{sec:2-approximation}
Given a polygonal curve $\pi$ in $\R^d$, let $H^*$ be a $d$-cube with centre at $p^*$ and radius $r^*$ that has a maximum packedness value. Our approximation algorithm builds on two observation. The first observation is that given a center $p \in \R^d$ one can in $O(dn \log n)$ time find, of all possible $d$-cubes centered at $p$, the $d$-cube that has the largest packedness value. The second observation is that there exists a $d$-cube centered at a vertex of $\pi$ that has a packedness value that is at least half the packedness value of $H^*$.

Before we present the algorithm we need some notations. Let $H^p_r$ be the $d$-cube $H$, scaled with $p$ as center and such that its radius is $r$. Fix a point $p$ in $\R^d$, and consider $\Psi$ as a function of $r$. More formally, let $\psi_p(r) = \Psi(H_r^p)$. Gudmundsson~\etal~\cite{gks-ahct-13} showed properties of $\psi_p(r)$ that we generalize to $\R^d$ and restate as:

\begin{lemma} \label{lemma:hyperbolic function}
 The function $\psi_p(r)$ is a piecewise hyperbolic function. The pieces of $\psi_p(r)$ are of the form $a(1/r)+ b$, for $a, b \in \R$, and the break points of $\psi_p(r)$ correspond to $d$-cubes $H$ where: (i) a
vertex of $\pi$ lies on a $(d-1)$-face of $H$, or (ii) a $(d-2)$-face (in $\R^3$, an edge) of $H$ intersects an edge of $\pi$.
\end{lemma}

As a corollary we get:
\begin{corollary} \label{obs:extreme-radius}
Let $r_1, r_2$ be the radii of two consecutive break points of $\psi_p(r)$, where $r_2 > r_1$. It holds that $\max_{r\in[r_1,r_2]} \psi_p(r) = \max \{\psi_p(r_1), \psi_p(r_2)\}$, that is, the maximum value is obtained either at $r_1$ or at $r_2$.
\end{corollary}
\begin{proof}
According to Lemma~\ref{lemma:hyperbolic function} the function $\psi_p(r)$ is a hyperbolic function in the range $r\in[r_1,r_2]$ of the form $a(1/r)+ b$ with the derivative $-a/r^2$. This implies that $\psi_p(r)$ is a monotonically decreasing or monotonically increasing function in $[r_1,r_2]$. As a result the maximum value of $\psi_p(r)$ is attained either at $r_1$ or at $r_2$.
\end{proof}


Next we state the algorithm for the first observation. The general idea is to use plane-sweep, scaling $d$-cube $H$ with centre at $p$ by increasing its radius from $0$ to $\infty$. The $(d-1)$-faces of $H$ are bounded by $2d$ hyperplanes in $\R^d$. When $H$ expands from $p$, it can first meet a segment of $\pi$ in one of two ways: (i) a vertex of the segment lies on one of $H$'s $(d-1)$-face, or (ii) an interior point of the segment lies on a $(d-2)$-face of $H$. For the first case, the new segment can change to intersect a different $(d-1)$-face of $H$ at most $d-1$ times, depending on its relative position to the center of $H$ and its components in all $d$ dimensions. 

Similarly for a segment of the second case, it can change to intersect a different $(d-1)$-face of $H$ $O(d)$ times. Thus each segment has $O(d)$ event points and there are $O(dn)$ events in total. Sort the events by their radii $r_1,\ldots,r_m (m=O(dn))$ in increasing order. Perform the sweep by increasing the radius $r$ starting at $r=0$ and continue until all events have been encountered.

Recall that $\Upsilon(H)$ is the total length of the trajectory $\pi$ inside $H$. For each $r_i$, $1\leqslant i\leqslant m$, we can compute $\psi_p=\Upsilon(r_i)/r$ in time $O(dn)$. For two consecutive radii $r_i$ and $r_{i+1}$, $\Upsilon(H^p_{r_i})$ and $\Upsilon(H^p_{r_{i+1}})$ can differ in one of three ways. First, $H^p_{r_{i+1}}$ may include a vertex not in $H^p_{r_i}$, in which case the set of contributing edges may increase by up to two. Second, $H^p_{r_{i+1}}$ may intersect an edge not in $H^p_{r_{i}}$. Finally, an edge in $H^p_{r_i}$ may intersect a different $(d-1)$-face. We can compute a function $\Delta(r_i,r_{i+1})$ that describes these changes in constant time. We then have $\Upsilon(H^p_{r_{i+1}}) = \Upsilon(H^p_{r_i}) + \Delta(r_i,r_{i+1})$,
and we can compute $\Upsilon(H^p_{r_{i+1}})$ from $\Upsilon(H^p_{r_i})$ in constant time (in $\R^2$ similar to~\cite{msw-aotp-96}). Apart from sorting the event points, we compute $\psi_p(r)$ for every $r_i$, $1\leqslant i\leqslant m$, in $O(dn)$ time. We return radius $\argmax_{r_1\leqslant r_i\leqslant r_m}\psi_p(r_i)$ as the result. Hence, the total running time is $O(dn \log n)$.

Note that the break points of $\psi_p(r)$ are the event points. The correctness follows immediately from Corollary~\ref{obs:extreme-radius} which tells us that we only need to consider the set of event points. To summarise we get:

\begin{lemma} \label{lemma:fixed-centre}
Given a point $p$ in $\R^d$ one can in $O(dn \log n)$ time determine the radius $r>0$ such that $\Psi(H^p_r)=\max_{r'>0} \psi_p(r')$. 
\end{lemma}

Now we are ready to prove the second observation.

\begin{lemma}\label{segIIMono}

Consider the function $\psi_p(r)$ for a single segment $s$, i.e. $\psi_p(r)=\frac{|s \cap H^p_r|}{r}$. If the first point on $s$ encountered by $H$ is an interior point of $s$ then the function is non-decreasing from $r=0$ until $H^p_r$ encounters a vertex of $s$.     
\end{lemma}
\begin{proof}
The function is zero until an interior point on $s$ is encountered. After encountering the interior point 
and before encountering a vertex of~$s$, the segment $|s \cap H^p_r|$ is a chord between two boundary points of $H^p_r$. Suppose we normalise the size of the $d$-cube $H^p_r$ to be unit-sized. Then the length of the chord is normalised to $\frac{|s \cap H^p_r|}{r} = \psi_p(r)$. Before normalisation, the segment $|s \cap H^p_r|$ had fixed gradient, and had fixed orthogonal distance to the center $p$. After normalisation, the chord has fixed gradient and has decreasing distance to the center $p$. Therefore its length~$\psi_p(r)$ is non-decreasing as it approaches the diameter of~$H^p_r$.
\end{proof}


\begin{lemma}\label{2ApprxAlgo}
There exists a $d$-cube $H$ with center at a vertex of $\pi$ such that $\Psi(H) \geq \frac{1}{2} \cdot \Psi(H^*)$, where $H^*$ is the $d$-cube having the highest packedness value for $\pi$.
\end{lemma}
\begin{proof}
Consider $H^*$. We will construct a $d$-cube $H$ that is centered at a vertex $p$ of $\pi$ and contains $H^*$. We will then prove that $H$ has packedness value at least $\frac{1}{2} \cdot \Psi(H^*)$, which would prove the theorem. To construct $H$, we consider two cases: 

\noindent {\bf Case 1:} The square~$H^*$ does not contain a vertex of $\pi$, see Fig.~\ref{expandTheSquare}(a). Scale $H^*$ until its boundary hits a vertex. Let $H_1$ denote the $d$-cube obtained from the scaling and let $v$ be the vertex on the $(d-1)$-face of $H_1$. According to Lemma~\ref{segIIMono}, we know that $\Psi(H_1) \geq \Psi(H^*)$.

Let $H_2$ be the $d$-cube centered at $v$ with radius twice the radius of $H_1$, as illustrated in Fig.~\ref{expandTheSquare}(a). Clearly $H_2$ contains $H_1 \cap \pi$, so $\Psi(H_2) \geq \frac {1}{2} \Psi(H_1) \geq \frac {1}{2} \Psi(H^*)$, as required.

\noindent {\bf Case 2:} The $d$-cube $H^*$ contains one or more vertices, see Fig.~\ref{expandTheSquare}(b). Let $v$ be a vertex inside $H^*$. Let $H_2$ be the $d$-cube with center at $v$ and radius twice that of $H^*$. Again, we have $H_2$ completely contains $H^*\cap \pi$, so $\Psi(H_2) \geq \frac{1}{2}\Psi(H^*)$, as required.

In both cases, we have constructed a $d$-cube $H_2$ centered at a vertex of $\pi$ for which $\Psi(H_2) \geq \frac{1}{2} \Psi(H^*)$, which proves the lemma.
\end{proof}
\begin{figure}[bth]
\centering
  \includegraphics[width=11cm]{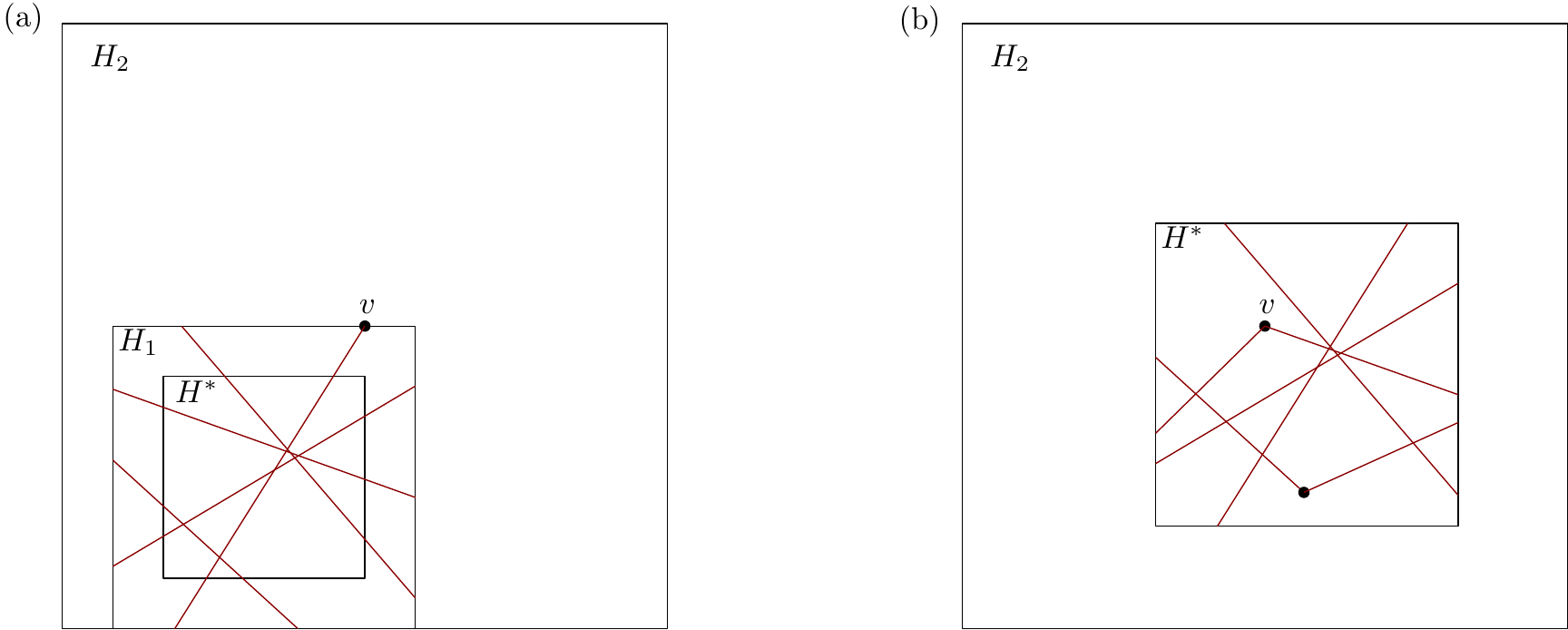}
\caption{Illustrating the two cases in the proof of Lemma~\ref{2ApprxAlgo}: Case 1 in (a) and Case 2 in (b).}
\label{expandTheSquare}
\end{figure}


\begin{table}[tbh]
    \centering
    \begin{tabular}{|c||c|c|c|c|c|c|}
    \hline
    \bf Dataset & \#Curves & MaxCurveSize & Min & Max &  Avg & Avg $c/n$ \\ \hline \hline
     Vessel-Y & 187  & 320& 2.37 & 14.28 &  3.03 & 0.022\\ \hline
     Hurdat &1785 & 133 & 2 & 16.58 &  3.24 & 0.154 \\ \hline
     Pen & 2858 & 182 & 4.07 & 20.82 & 8.79 & 0.073 \\ \hline
     Bats &545 & 736 & 1.08&29.52&3.60 & 0.0625 \\ \hline
     Bus &148 &1012 & 3.21 & 34.99 &  14.70 & 0.052 \\ \hline
     Vessel-M & 103 & 143 &1.04&46.19 & 4.66 & 0.272 \\ \hline
     Basketball &20780 & 138 & 2.00 & 48.65 &  3.95 & 0.092 \\ \hline
     Football &18028 & 853 & 2.12 & 48.87 & 7.66 & 0.045 \\ \hline
     Truck &276  & 983 & 5.32 & 110.44 & 25.48 & 0.079 \\ \hline
     Buffalo & 163 & 479 &  1.17 & 254.14  & 68.42 & 0.505 \\ \hline
     Pigeon & 131 & 1504 & 3.64 & 275.18 & 90.93 & 0.12 \\ \hline
     Geolife & 1000 & 64390 & 1.02 & 858.19  & 23.31 & 0.057 \\ \hline
     Gull & 241 &3237 & 1.02 & 1082.20&139.50& 0.478 \\ \hline
     Cats  & 152 &2257 & 6.04 & 1122.77 & 207.86 & 0.655 \\ \hline
     Seabirds & 63  &2970 &5.59&1803.72&825.57& 0.388\\ \hline
     Taxi & 1000  & 115732 & 2.20 & 4255.23 & 55.38 & 0.313\\ \hline
    \end{tabular}
    \caption{The table lists 16 real-world data sets. The second and third columns shows the number of curves and the maximum complexity of a curve in the set. The following three columns lists the minimum, maximum and average approximate packedness values. The rightmost column states the average ratio between $c$ and $n$ for the data sets.}
    \label{tab:ExperimentsData}
\end{table}

\begin{table*}[bth] 
	\centering
{ \setlength{\tabcolsep}{0.5em}
	\begin{tabular}{l|r|r|r|l}
		\hline
		Data Set & $n$ & $d$ & $\#$vertices & Trajectory Description \\
		\hline
Vessel-M~\cite{vessel05} 	& $106$ 	& $2$ 	& $23.0$ 	 	& Mississippi river shipping vessels Shipboard AIS.\\
Pigeon~\cite{pigeon16} 		& $131$ 	& $2$ 	& $970.0$ 	 	& Homing Pigeons (release sites to home site).\\
Seabird~\cite{masked17} 	& $134$ 	& $2$ 	& $3175.8$ 	 	& GPS of Masked Boobies in Gulf of Mexico.\\
Bus~\cite{truckbus05} 		& $148$ 	& $2$ 	& $446.6$ 	 	& GPS of School buses.\\
Cats~\cite{cats16} 			& $154$ 	& $2$ 	& $526.1$ 	 	& Pet house cats GPS in Raleigh-Durham, NC, USA.\\
\hline
Buffalo~\cite{kruger09} 	& $165$ 	& $2$ 	& $161.3$ 	 	& Radio-collared Kruger Buffalo, South Africa.\\
Vessel-Y~\cite{vessel05} 	& $187$ 	& $2$ 	& $155.2$ 	 	& Yangtze river shipping Vessels Shipboard AIS.\\
Gulls~\cite{gulls15} 		& $253$ 	& $2$ 	& $602.1$ 	 	& Black-backed gulls GPS (Finland to Africa).\\
Truck~\cite{truckbus05} 	& $276$ 	& $2$ 	& $406.5$ 	 	& GPS of 50 concrete trucks in Athens, Greece.\\
Bats~\cite{bats15} 			& $545$ 	& $2$ 	& $44.1$ 	 	& Video-grammetry of Daubenton trawling bats.\\
\hline
Hurdat2~\cite{hurdat217} 	& $1788$ 	& $2$ 	& $27.7$ 	 	& Atlantic tropical cyclone and sub-cyclone paths.\\
Pen~\cite{pentip06} 		& $2858$ 	& $2$ 	& $119.8$ 	 	& Pen tip characters on a WACOM tablet.\\
Football~\cite{soccer15} 	& $18034$ 	& $2$ 	& $203.4$ 	 	& European football player (team ball-possession).\\
Geolife~\cite{geo2012} 		& $18670$ 	& $2$ 	& $1332.5$ 	 	& People movement, mostly in Beijing, China.\\
Basketball~\cite{NBA16} 	& $20780$ 	& $3$ 	& $44.1$ 	 	& NBA basketball three-point shots-on-net.\\
Taxi~\cite{taxiA11,taxiB10} & $180736$ 	& $2$ 	& $343.0$ 	 	& $10$,$357$ Partitioned Beijing taxi trajectories.\\
\hline
	\end{tabular}
}
	\caption{Real data sets, showing number of input trajectories $n$, dimensions $d$, average number of simplified vertices per trajectory, and a description.}
	\label{tab:NNRes}
\end{table*}

\subsection{Experimental results} \label{sec:experiments}
We implemented the above algorithm to test the approximate packedness of real-world trajectory data. We ran the algorithm on 16 data sets. The data sets were kindly provided to us by the authors of~\cite{ghps-pissf-20}. Table~\ref{tab:NNRes} summarises the data sets and is taken from~\cite{ghps-pissf-20}. The minimum/maximum/average (approximate) packedness values and the ratio between $c$ and $n$ for each dataset are listed in Table~\ref{tab:ExperimentsData}. Both the Geolife dataset and the Taxi dataset consist of over 20k trajectories, many of which are very large. For the experiments we randomly sampled 1,000 trajectories from each of these sets.

Although these are only sixteen data sets, it is clear that the notion of $c$-packedness is a reasonable model for many real-world data sets. For example, the maximal packedness value for all trajectories in all the first eight data sets is less than 50 and the average (approximate) packedness value is below 15. Looking at the ratio between $c$ and $n$, we can see that for many data sets the value of $c$ is considerably smaller than $n$. 

Consider the task of computing the continuous Fr\'echet distance between two trajectories. For two trajectories of complexity $n$, computing the distance will require $O^*(n^2)$ time  (even for an $O(1)$-approximation) while a $(1+\varepsilon)$-approximation can be obtained for $c$-packed trajectories in $O^*(cn)$ time.  Thus the algorithm by Driemel~\etal~\cite{DBLP:journals/dcg/DriemelHW12} for $c$-packed curves is likely to be more efficient than the general algorithm for these data sets. 



\section{A fast \texorpdfstring{$(6+\varepsilon)$}{(6+epsilon)}-approximation algorithm} \label{sec:6-approximation}

In this section we will take a different approach to Section~\ref{sec:2-approximation} to yield an algorithm that considers a linear number of squares rather than a quadratic number of squares. First we will identify a set $\mathcal{S}$ containing a linear number of squares that will include a square having a high packedness value (Section~\ref{ssec:linear-number}). Then we will build a multi-level data structure (Section~\ref{ssec:data-structure}) on $\pi$ such that given a square $S \in \mathcal{S}$ it can quickly approximate $|S\cap \pi|$.

\subsection{Linear number of good squares} \label{ssec:linear-number}
To prove that it suffices to consider a linear number of squares we will use the well-known Well-Separated Pair Decomposition (WSPD) by Callahan and Kosaraju~\cite{ck-dmpsa-95}. 

Let $A$ and $B$ be two finite sets of points in $\R^d$ and let $s > 0$ be a real number. We say that $A$ and $B$ are well-separated with respect to $s$, if there exist two disjoint balls $C_A$ and $C_B$, such that (1) $C_A$ and $C_B$ have the same radius, (2) $C_A$ contains the bounding box of $A$ and $C_B$ contains the bounding box of $B$, and (3) the distance between $C_A$ and $C_B$ is at least $s$ times the radius of $C_A$ and $C_B$. The real number $s$ is called the separation ratio.

\begin{lemma}\label{lemma:properties-of-ws-pairs}
Let $s > 0$ be a real number, let $A$ and $B$ be two sets in $\R^d$ that are well-separated with respect to $s$, let $a$ and $a'$ be two points in $A$, and let $b$ and $b'$ be two points in $B$. Then (1) $|aa'|\leq (2/s)\cdot |ab|$, and (2) $|a'b'| \leq  (1 + 4/s)\cdot |ab|$.
\end{lemma}

\begin{definition}
Let $S$ be a set of $n$ points in $\R^d$, and let $s > 0$ be a real number. A well-separated pair
decomposition (WSPD) for $S$, with respect to $s$, is a sequence $\{A_1, B_1\}, \ldots ,$  $ \{A_m, B_m\}$ of pairs of non-empty subsets of $S$, for some integer $m$, such that:
\begin{enumerate}
    \item for each $i$ with $1 \leq i \leq m$, $A_i$ and $B_i$ are well-separated with respect to $s$, and
    \item for any two distinct points $p$ and $q$ of $S$, there is exactly one index $i$ with $1 \leq i \leq m$, such that $p \in A_i$ and $q \in B_i$, or $p \in B_i$ and $q \in A_i$. The integer $m$ is called the size of the WSPD.
\end{enumerate}
\end{definition}

\begin{lemma} \label{lemma:WSPD} (Callahan and Kosaraju~\cite{ck-dmpsa-95}) Given a set $V$ of $n$ points in $\R^d$, and given a real number $s > 0$, a well-separated pair decomposition for $V$, with separation ratio $s$, consisting of $O(s^dn)$ pairs, can be computed in $O(n \log n + s^dn)$ time.
\end{lemma}

Now we are ready to construct a set $\mathcal{S}$ of squares. Compute a well-separated pair decomposition $W=\{(A_1,B_1), \ldots , (A_m,B_m)\}$ with separation constant $s=720/\varepsilon$ for the vertex set of $\pi$.
For every well-separated pair $(A_i,B_i)\in W$, $1\leq i\leq k$, construct two squares that will be added  to $\mathcal{S}$ as follows:

Pick an arbitrary point $a\in A_i$ and an arbitrary point $b\in B_i$. Construct one square with center at $a$ and radius $r$, and one square with center at $b$ and radius $r$, where $r=\max\{|a.x-b.x|,|a.y-b.y|\}+\varepsilon/120 \cdot |ab|$. The two squares are added to $\mathcal{S}$.

It follows immediately from Lemma~\ref{lemma:WSPD} that the number of squares in $\mathcal{S}$ is $O(n/\varepsilon^2)$ and that one can construct $\mathcal{S}$ in $O(n \log n+n/\varepsilon^2)$ time.

To prove the approximation factor of the algorithm we will first need the following technical lemma. 

\begin{lemma} \label{lemma:vertexOnBoundaryLemma}
 Let $H_{r_1}^p$ and $H_{r_2}^p$, with $r_2 > r_1$, be two squares with centre at $p$ such that $H_{r_2}^p \setminus H_{r_1}^p$ contains no vertices of $\pi$ in its interior. For any value $r_x$, with $r_1\leqslant r_x\leqslant r_2$, it holds that $\psi_p(r_x) \leq \psi_p(r_1) + 2\cdot \psi_p(r_2)$.
\end{lemma}

\begin{figure}[h!]
\centering
\includegraphics[width=\textwidth]{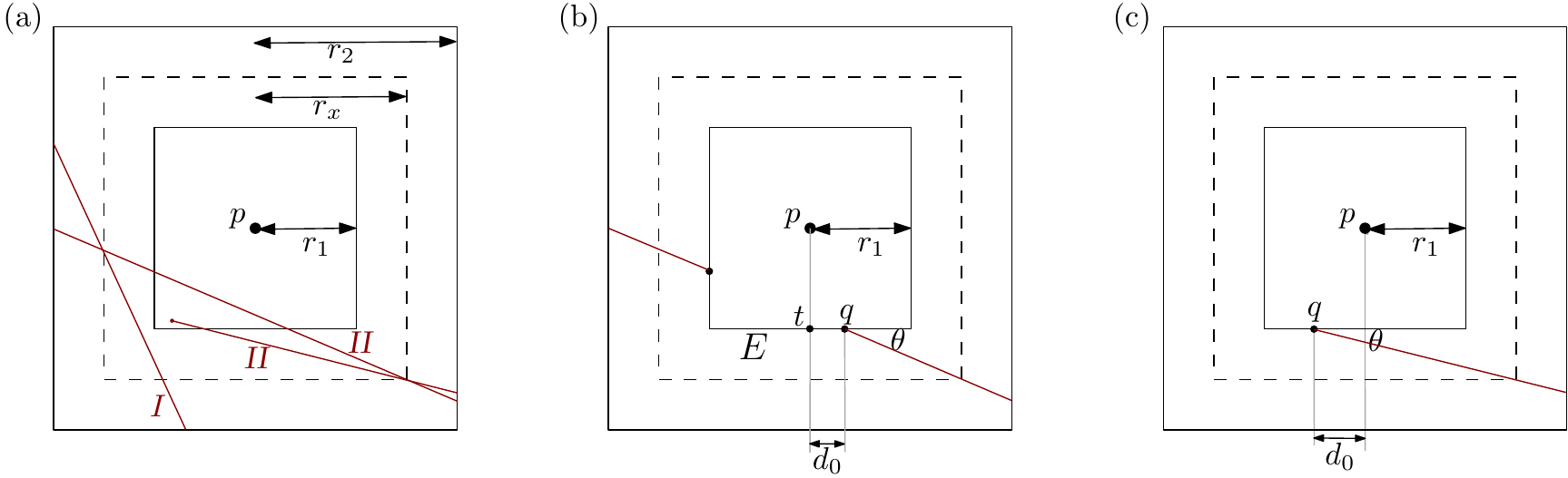}
\caption{(a) An illustration of the proof of Lemma~\ref{lemma:vertexOnBoundaryLemma} and the two types of segments that are considered. (b) Showing case 1, and (c) case 2 of of Type II segments.}
\label{fig:between2Vertices}
\end{figure}
\begin{proof}
Let
\begin{align*}
    \psi_p(r_1) &= \frac{\Upsilon(H^p_{r_1})}{r_1},\\ 
    \psi_p(r_2) &= \frac{\Upsilon(H^p_{r_1})}{r_2}+\frac{M_2}{r_2}, \text{ and} \\
    \psi_p(r_x) &=
    \frac{\Upsilon(H^p_{r_1})}{r_x}+\frac{M_x}{r_x},
\end{align*}
where $M_2=\Upsilon(H^p_{r_2})- \Upsilon(H^p_{r_1})$ and $M_x=\Upsilon(H^p_{r_x})- \Upsilon(H^p_{r_1})$.

It is clear that $\frac{\Upsilon(H^p_{r_1})}{r_x}<\frac{\Upsilon(H^p_{r_1})}{r_1}$, so the first term of $\psi_p(r_x)$ is bounded by $\psi_p(r_1)$. Next we consider the second term.

There are no vertices of $\pi$ inside $H^p_{r_2}\setminus H^p_{r_1}$. These segments either cross $H^p_{r_2}\setminus H^p_{r_1}$ and have no intersection with $H^p_{r_1}$, like segment I in Fig.~\ref{fig:between2Vertices}(a), or intersect $H^p_{r_1}$ and cross $H^p_{r_2}\setminus H^p_{r_1}$, like segment II in Fig.~\ref{fig:between2Vertices}(a).

Let $l_x$ denote a segment's contribution to $M_x$. Let us consider how $l_x/r_x$ changes when $r_x$ grows from $r_1$ to $r_2$. 

\begin{enumerate}
    \item Segment of Type I: We know from Lemma~\ref{segIIMono} that $l_x/r_x$ is non-decreasing as $r_x$ grows from $r_1$ to $r_2$. 
    
    \item Segment of Type II: Consider a segment $s$ of Type II and the subsegment (or two subsegments) $s\cap (H^p_{r_1}\setminus H^p_{r_2})$. A subsegment $s'$ of $s$ has one endpoint $q$ on the boundary of $H^p_{r_1}$ and one endpoint on the boundary of $H^p_{r_2}$. Let $E$ be the side of $H^p_{r_1}$ containing $q$. Let $t$ be the middle point of $E$, let $d_0=|qt|$ and let $\theta$ be the acute angle between $s'$ and $E$. There are two cases.
    \begin{itemize}
    \item The subsegment $s'$ does not cross line $pt$ in region $H^p_{r_2}\setminus H^p_{r_1}$, as in Figure~\ref{fig:between2Vertices}(b).
    Let $r_x'$ be the radius when $s'$ crosses a corner of $H^p_{r_x}$. When $r_x<r_x'$,
    \begin{align*}
    \frac{l_x}{r_x}=\frac{\frac{\Delta r}{\sin\theta}}{r_x}=\frac{1}{\sin\theta}\cdot\frac{\Delta r}{r_1+\Delta r}
    \end{align*}
    increases strictly. When $r_x\geq r_x'$,
    \begin{equation*}
    \frac{l_x}{r_x}=\frac{\frac{r_x-d_0}{\cos\theta}}{r_x}=\frac{1}{\cos\theta}\cdot(1-\frac{d_0}{r_x})
    \end{equation*}
    continues to increase strictly, so $l_x/r_x\leqslant l_2/r_2$.
    
    \item The subsegment $s'$ crosses line $pt$ in region $H^p_{r_2}\setminus H^p_{r_1}$, as shown in Figure~\ref{fig:between2Vertices}(c). $r_x'$ is defined as above. When $r_x<r_x'$, $l_x/r_x$ increases strictly. When $r_x\geq r_x'$, 
    \begin{equation*}
    \frac{l_x}{r_x}=\frac{1}{\cos\theta}\cdot\frac{r_x+d_0}{r_x}=\frac{1}{\cos\theta}\cdot(1+\frac{d_0}{r_x})
    \end{equation*}
begins to decrease. However, since $d_0\leqslant r_1\leqslant r_x\leqslant r_2$,
\begin{equation*}
    \frac{1}{\cos\theta}\cdot\frac{r_x+d_0}{r_x}<\frac{1}{\cos\theta}\cdot 2<2\cdot\frac{1}{\cos\theta}\cdot (1+\frac{d_0}{r_2})=2\cdot\frac{l_2}{r_2}.
\end{equation*}
\end{itemize}

\end{enumerate}    
    For all cases, $\frac{l_x}{r_x}\leqslant 2\cdot\frac{l_2}{r_2}$. Thus $\frac{M_x}{r_x}\leqslant 2\cdot\frac{M_2}{r_2}$. We get
\begin{equation*}
    \psi_p(r_x)< \psi_p(r_1)+2\cdot\frac{M_2}{r_2}\leqslant\psi_p(r_1)+ 2\psi_p(r_2).
\end{equation*}
\end{proof}

Due to Lemma~\ref{2ApprxAlgo}, it suffices to consider squares with center at a vertex of $\pi$ to obtain a $2$-approximation. Combining this with Lemma~\ref{lemma:vertexOnBoundaryLemma}, it suffices to consider squares with center at a vertex of $\pi$ and a vertex of $\pi$ on its boundary to obtain a $6$-approximation. Using the WSPD argument we have reduced our set of squares to a linear number of squares and we will now argue that $\mathcal{S}$ must contain a square that has a high packedness factor. Let $H^*$ be a square with a maximum packedness value of $\pi$.

\begin{lemma} \label{lemma:linear-number-of-squares}
 There exists a square $S \in \mathcal{S}$ such that $\Psi(S) \geqslant \frac{1}{(6+\varepsilon/8)} \cdot \Psi(H^*)$. 
\end{lemma}
\begin{proof}
From Lemmas~\ref{2ApprxAlgo} and~\ref{lemma:vertexOnBoundaryLemma} we know that there exists a square $H$ with centre at a vertex $p$ of $\pi$ and whose boundary contains a vertex $q$ of $\pi$ such that $\Psi(H)\geq \frac{1}{6} \cdot \Psi(H^*).$ According to the construction of $\mathcal{S}$ there exists a square $S\in\mathcal{S}$ such that $S$ has its centre at a point $a\in A_i$ and has radius $r=\max\{|a.x-b.x|,|a.y-b.y|\}+\varepsilon/120 \cdot |ab|$, where $b$ is a point in $B_i$. 

By Lemma \ref{lemma:properties-of-ws-pairs}, we have $|ap|\leq \varepsilon/360\cdot |ab|$ and $|bq|\leq \varepsilon/360\cdot |ab|$. If $r_H$ is the radius of~$H$, then $r_H=\max\{|p.x-q.x|,|p.y-q.y|\}\leq \max\{|a.x-b.x|,|a.y-b.y|\}+|ap|+|bq|\leq \max\{|a.x-b.x|,|a.y-b.y|\}+\varepsilon/120 \cdot |ab| - |ap| = r - |ap|$. But $p$ is at most $|ap|$ away from $a$ in both the $x$ and $y$ directions, so $H$ must be entirely contained inside $S$. So $\Upsilon(S)\geq \Upsilon(H)$.

Next, we show $S$ is not too much larger than $H$:
\begin{equation*}
\begin{split}
r & =  \max\{|a.x-b.x|,|a.y-b.y|\}+\frac{\varepsilon}{120} \cdot |ab|  
   \leqslant r_H+|ap|+|bq|+\frac{\varepsilon}{120} \cdot |ab|  
   \leqslant r_H+\frac{\varepsilon}{72} \cdot |ab| \\
   & \leqslant 
   r_H+\frac{\varepsilon}{72}\cdot(1+\frac{\varepsilon}{180}) \cdot |pq| 
   \leqslant r_H+\frac{\varepsilon}{36\sqrt{2}}\cdot(1+\frac{\varepsilon}{180}) \cdot r_H  \leqslant (1+\frac{\varepsilon}{48}) \cdot r_H.
\end{split}
\end{equation*}

Putting this all together yields: 
$$\Psi(S) 
    = \Upsilon(S)/r 
    \geq \frac{1}{(1+\varepsilon/48)  } \cdot\Upsilon(H)/r_H
    = \frac{1}{1+\varepsilon/48} \cdot \Psi(H) 
    \geq \frac{1}{6+\varepsilon/8} \cdot \Psi(H^*),$$
which completes the lemma.
\end{proof}

\subsection{Data structure} \label{ssec:data-structure}

The aim of this section is to develop an efficient data structure on $\pi$ such that queried with an axis-aligned square $S \in \mathcal{S}$ the data structure returns an approximation of $|S\cap \pi|$.

The general idea of the multi-level data structure is that the first level is a modified 1D segment tree, similar to the hereditary segment tree~\cite{c-rcsi-86}. We partition the set of $\pi$'s segments into four sets depending on their slope;  $(-\infty,-1)$, $[-1,0)$, $[0,1)$ and $[1,\infty)$. In the rest of this section we will describe the data structure for the set of segments with slope in $[0,1)$. The remaining three sets are handled symmetrically. 

\subsubsection{Modified 1D segment tree}
The description of the segment tree follows the description in~\cite{bkos-cgta-08}. Let $L$ be the set of line segments in $\pi$. For the purpose of the 1D segment tree, we can view $L$ as a set of intervals on the line. Let $p_1, \ldots , p_m$ be the list of distinct interval endpoints, sorted from left to right. Consider the partitioning of the real line induced by those points. The regions of this partitioning are called \emph{elementary intervals}. Thus, the elementary intervals are, from left to right:
$(-\infty ,p_{1}),[p_{1},p_{1}],(p_{1},p_{2}),[p_{2},p_{2}], \ldots , (p_{m-1},p_{m}),[p_{m},p_{m}],(p_{m},+\infty )$.

Given a set $I$ of intervals, or segments, a segment tree $\mathcal{T}$ for $I$ is structured as follows:
\begin{enumerate}
    \item $\mathcal{T}$ is a binary tree.
    \item Its leaves correspond to the elementary intervals induced by the endpoints in $I$. The elementary interval corresponding to a leaf $v$ is denoted $Int(v)$.
    \item The internal nodes of $\mathcal{T}$ correspond to intervals that are the union of elementary intervals: the interval $Int(N)$ corresponding to an internal node $N$ is the union of the intervals corresponding to the leaves of the tree rooted at $N$. That implies that $Int(N)$ is the union of the intervals of its two children.
    \item Each node or leaf $v$ in $\mathcal{T}$ stores the interval $Int(v)$ and a set of intervals, in some data structure. This canonical subset of node $v$ contains the intervals $[x, x']$ from $I$ such that $[x, x']$ contains $Int(v)$ and does not contain $Int(parent(v))$. That is, each node in $\mathcal{T}$ stores the set of segments $F(v)$ that span through its interval, but do not span through the interval of its parent.
\end{enumerate}
The 1D segment tree can be built in $O(n \log n)$ time, using $O(n \log n)$ space and point stabbing queries can be answered in $O(\log n+k)$ time, where $k$ is the number of segments intersecting the query point.

We make one minor change to $\mathcal{T}$ that will increase the space usage to $O(n \log^2 n)$ but it will allow us to speed up interval stabbing queries. Each internal node $v$ store, apart from the set $F(v)$, all the segments stored in the subtree rooted at $v$, including $F(v)$. We denote this set by $L(v)$.

The main benefit of this minor modification is that when an interval stabbing query is performed only $O(\log n)$ canonical subsets are required to identify all the segments intersecting the interval. Next we show how to build associated data structures for $L(v)$ and $F(v)$ for each internal node $v$ in $\mathcal{T}$.


\subsubsection{Three associated data structures}
Consider querying the segment tree $\mathcal{T}$ with a square $S\in \mathcal{S}$. There are three different cases that can occur, and for each of these cases we will build an associated data structure. That is, each internal node will have three types of associated data structures. Consider a query $S=[x,x']\times [y,y']$ and let $\mu_l$ and $\mu_r$ be the leaf nodes in $\mathcal{T}$ where the search for the boundary values $x$ and $x'$ end. See Figure~\ref{fig:ThreeADS} for an illustration of the search and the three cases. An internal node $v$ is one of the following types:
\begin{description}
  \item{{\bf Type A:}} if $Int(v) \subseteq [x,x']$, 
  \item{{\bf Type B:}} if $Int(v) \cap [x,x'] \neq \emptyset$, $Int(v) \nsubseteq [x,x']$ and $[x,x'] \nsubseteq Int(v)$, or 
  \item{{\bf Type C:}} if $[x,x'] \subset Int(v)$. 
\end{description}

\begin{figure}[htb]
\centering
\includegraphics[width=\textwidth]{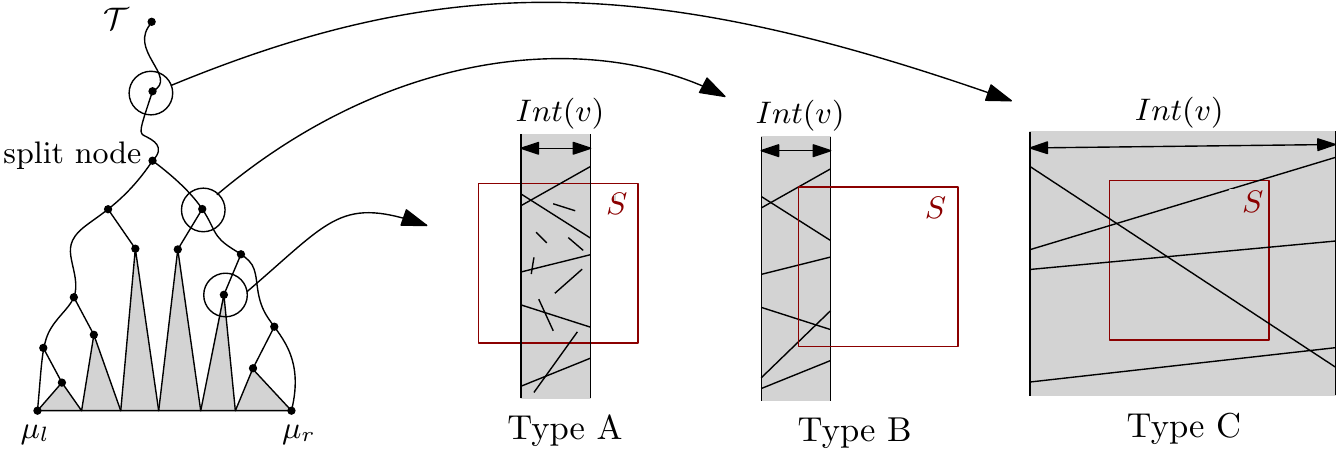}
\caption{The primary tree $\mathcal{T}$, and the three types of nodes in $\mathcal{T}$ that can be encountered during a query.}
\label{fig:ThreeADS}
\end{figure}

\paragraph*{Associated data structure for Type A nodes:} For a Type A node we need to compute the length of all segments stored in the subtree with root $v$ in the $y$-interval $[y,y']$. Let  $s_1, \ldots, s_m$ be the set of $m$ segments stored in $L(v)$, and let $Y=\langle y_1, \ldots, y_{2m} \rangle$ denote the $y$-coordinates of the endpoints of the segments in $L(v)$ ordered from bottom-to-top. To simplify the description we assume that the values are distinct.

Let $\delta(y)$ denote the total length of the segments in $L(v)$ below $y$.  
For two consecutive $y$-values $y_i$ and $y_{i+1}$, the set of edges contributing to $\delta(y_i)$ and $\delta(y_{i+1})$ has increased or decreased by one. So, we can compute a function $\Delta(y_i,y_{i+1})$  that describes these changes in constant time. We then have $\delta(y_{i+1}) = \delta(y_i) + \Delta(y_i,y_{i+1})$, and thus we can compute $\delta(y_{i+1})$ from $\delta(y_i)$ in constant time after sorting the events.  Hence, we can compute all the $\delta(y_i)$-values and all the $\Delta(y_i,y_{i+1})$ in time $O(m \log m)$. 

Given a $y$-value $y'$ one can compute $\delta(y')$ as $\delta(y_i)+\frac{y'-y_i}{y_{i+1}-y_i} \cdot \Delta(y_i,y_{i+1})$, where $y_i$ is the largest $y$-value in $Y$ smaller than $y'$. Hence, our associated data structure for Type A nodes is a binary tree with respect to the values in $Y$, where each leaf stores the value $y_i$, $\delta(y_i)$ and $\Delta(y_i,y_{i+1})$. The tree can be computed in $O(m \log m)$ time using linear space, and can answer queries in $O(\log m)$ time. 

\begin{lemma} \label{lemma:Type-A}
The associated data structures for Type A nodes in $\mathcal{T}$ can be constructed in $O(n \log^2 n)$ time using $O(n\log^2 n)$ space. Given a query square $S$ for an associated data structure of Type A stored in an internal node $v$, the value $|L(v)\cap S \cap Int(v)|$ is returned in time $O(\log n)$.
\end{lemma}


\paragraph*{Associated data structure for Type B nodes:} 
The associated data structure for a Type B node is built to handle the case when the query square $S=[x,x']\times [y,y']$ intersects either the left boundary $(x_l)$ or the right boundary ($x_r$) of $Int(v)$, but not both. The two cases are symmetric and we will only describe the case when $S$ intersects the right boundary.  

The data structure returns a value $M$ that is an upper bound on the length of the segments of $F(v)$ within $S$ and a lower bound on the segments within $S^+$, where $S^+$ is a slightly expanded version of $S$. See Figure~\ref{fig:Type B}(c). Formally, $S^+=[x-\varepsilon/8 \cdot |x_r-x|,x'+\varepsilon/8 \cdot |x_r-x|] \times [y-\varepsilon/8 \cdot |x_r-x|,y'+\varepsilon/8 \cdot |x_r-x|]$. 

If $S^+$ spans $Int(v)$ then we need to use a binary tree on $F(v)$ to answer the query in logarithmic time, similar to the associated data structure for Type A nodes. 

If $S^+$ does not span $Int(v)$ then the query is performed on the Type B associated data structures. We first show how to construct these data structures, then we show how to handle the query. Recall that all the segments in $F(v)$ span the interval $Int(v)$. Let $s_1, \ldots , s_m$ be the set of $m$ segment in $F(v)$ and let $\mu(s_i)$ be the angle of inclination\footnote{$\mu(s_i)$ is the $\arctan$ of the slope in the interval $[0,1)$.} of $s_i$. The angle of inclination for segments with slope in the interval $[0,1)$ is in the interval $[0,\pi/4)$. Partition $F(v)$ into $\kappa_1$ sets $F_1(v), \ldots, F_{\kappa_1}(v)$ such that for any segment $s_i \in F_j(v)$ it holds that $(j-1)\cdot \frac{\pi}{4\kappa_1} \leq \mu(s_i) < j\cdot \frac{\pi}{4\kappa_1}$.

Consider one such partition $F_j(v)=\{s_1^j, \ldots , s^j_{m_j}\}$. Build a balanced binary search tree $T^j_r$ on the $y$-coordinates of the right endpoints of the segments in $F_j(v)$. The data structure can be constructed in $O(m_j \log m_j)$ time using linear space. Given a $y$-interval $\ell$ as a query, the number of right endpoints in $T^j_r$ within $\ell$ can be reported in time $O(\log m_j)$.  From the above description and the fact that $\sum_{v \in T} F(v)=O(n \log n)$ it immediately follows that the total construction time for all the Type B nodes is $O(n \log^2 n)$ and the total amount of space required is $O(n\log n)$. This completes the construction of the data structures on $F(v)$.

It remains to show how to handle a query, i.e. how to compute $M$. We focus first on computing an $M$ that upper bounds $|F(v) \cap S|$, and we later prove that $M$ lower bounds $|F(v) \cap S^+|$. 
There are two steps in computing $M$. The first step is to count the number of segments that intersect $S$. The second step is to multiply this count by the \emph{maximum} possible length of intersection between the segment and $S$. This would clearly yield an upper bound on $|F(v) \cap S|$. To obtain suitable maximum lengths in the second step, we need to subdivide the partitions $F_j$ further.


The right endpoints must lie in a $y$-interval given by $I_j=[y,y'+\bar{y}]$, where $y,y'$ are the $y$-coordinates of the bottom and top boundaries of $S$, and $\bar{y}=(x_r-x) \cdot \tan(\frac{j\cdot \pi}{4\kappa_1})$. Subdivide $I_j$ into three subintervals: $I^1_j=[y,y+\bar{y})$, $I^2_j=[y+\bar{y},y')$ and $I_j^3=[y',y'+\bar{y})$, see Fig.~\ref{fig:Type B}(b). Further subdivide $I^1_j$ and $I^3_j$ into $2\kappa_2$ subintervals of $\ell^1_1, \ldots , \ell^1_{\kappa_2}$ and $\ell^3_1, \ldots , \ell^3_{\kappa_2}$ of equal length. Hence we partitioned $I_j$ into a set $L_j$ of $2\kappa_2+1$ subintervals. 


Given these subdivisions $L_j$, our first step is to simply perform a range counting query in $T^j_r$ for each $\ell \in L_j$. Our second step is to multiply this count by the maximum length of intersection between $S$ and any segment in $F_j$ with its right endpoint in $\ell$. The product of these two values is clearly an upper bound on the length of intersection between $S$ and segments in $F_j$ with right endpoint in $\ell$. Finally, we sum over all subdivisions $\ell \in F_j$ and then over all partitions $F_j$ to obtain a value $M$ that upper bounds $|S \cap F(v)|$. The time required to handle a query is $O(\kappa_1 \cdot \kappa_2 \cdot \log m)$. It remains only to prove $M \leq |F(v)\cap S^+|$.



\begin{figure}[htb]
\centering
\includegraphics[width=12cm]{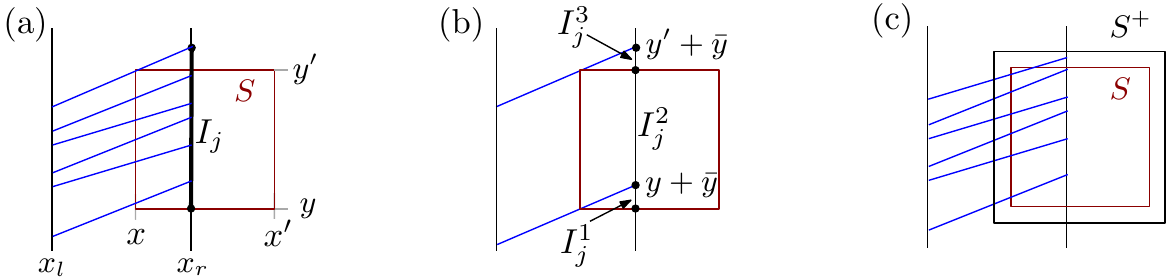}
\caption{(a) A query $S$ and the set $F(v)$. (b) Illustrating the three interval $I_j^1$, $I_j^2$ and $I_j^3$. (c) The expanded square $S^+$.}
\label{fig:Type B}
\end{figure}

By setting $\kappa_1=16\sqrt{2}/\varepsilon$ and $\kappa_2=16/\varepsilon$ we can prove the following.
\begin{lemma} \label{lemma:geometry}
$M\leq |F(v) \cap S^+|$ 
\end{lemma}
\begin{proof}
Consider an arbitrary interval $\ell \in L_j$. Let $A$ be a possible segment in $F_j(v)$ with right endpoint on $\ell$ that maximises $|A\cap S|$, and let $B$ be any segment in $F_j(v)$ with right endpoint on $\ell$.  It suffices to prove that $|A\cap S| \leq |B \cap S^+|$. We will have three cases depending on the position of $\ell$. Let $B'=B\cap S$ and let $B''=B\cap (S^+\setminus S)$.
\begin{itemize}
    \item $\ell \in I_j^1$: If the right endpoint of $B$ lies above $A$ then we can move $B$ vertically downward until the right endpoint coincides with the right endpoint of $A$. This will not increase $|B\cap S^+|$. 
    
    If $B$ has a smaller angle of inclination than $A$ then $|A \cap S|-|B'| \leq \frac{\sqrt{2} }{\kappa_1} \cdot |x_r-x|=\frac{\varepsilon}{16} \cdot |x_r-x|$. However, the length of $B''$ is at least $\varepsilon/8 \cdot |x_r-x|$, hence 
    $$|B \cap S^+| = |B'|+|B''| \geq |A \cap S|-\frac{\varepsilon}{16} \cdot |x_r-x|+\frac{\varepsilon}{8} \cdot |x_r-x| > |A \cap S|.$$
    
    If $B$ has a greater angle of inclination than $A$ then let $p$ be the left endpoint of $A\cap S$. Let $x_p$ be the $x$-coordinate of $p$ and let $p'$ be the point on $B$ with $x$-coordinate $x_p$. The distance between $p$ and $p'$ is bounded by $(\frac{1}{\kappa_2}+\frac{1}{\kappa_1}) \cdot |x_r-x| < \frac{\varepsilon}{8} \cdot |x_r-x|$, hence $p' \in S^+$, and it immediately follows that $|B\cap S^+| >  |A \cap S|$.
    \item $\ell \in I_j^2$: In this case the left endpoint of $A\cap S$ will be on the left boundary of $S$. Similarly, the left endpoint of $B'$ will lie on the left boundary of $S$. We know that $|A \cap S|-|B'| \leq \frac{\sqrt{2}}{\kappa_1} \cdot |x_r-x|=\frac{\varepsilon}{16} \cdot |x_r-x|$. However, the length of $B''$ is at least $\varepsilon/8 \cdot |x_r-x|$, hence 
    $$|B \cap S^+| \geq |B'| +|B''| \geq |A \cap S|-\frac{\varepsilon}{16} \cdot |x_r-x| +\varepsilon/8 \cdot |x_r-x| > |A \cap S|.$$
    \item $\ell \in I_J^3$: This case is very similar to the first case and left as an exercise for the reader. 
\end{itemize}
This shows that $|A\cap S| \leq |B \cap S^+|$ when $\kappa_1\leq 16\sqrt{2}/\varepsilon$ and $\kappa_2\leq 16/\varepsilon$, which proves the lemma.
\end{proof}

We summarise the associated data structures for Type B nodes with the following lemma.

\begin{lemma} \label{lemma:Type-B}
The associated data structure for Type B nodes can be constructed in $O(n \log^2 n)$ time using $O(n\log n)$ space. Given a query square $S=[x,x']\times [y,y']$ for an associated data structure of Type B stored in an internal node $v$, a real value $M(v)$ is returned in $O(\frac{1}{\varepsilon^2} \cdot \log n)$ time such that:
     $$|F(v) \cap S| \leq M(v) \leq |F(v) \cap S^+|,$$
where $S^+=([x-\varepsilon/8 \cdot w,x'+\varepsilon/8 \cdot w]) \times [y-\varepsilon/8 \cdot w,y'+\varepsilon/8 \cdot w]$ and $w$ is the width of $S\cap Int(v)$.
\end{lemma}


\paragraph*{Associated data structure for Type C nodes:} 
The associated data structure for Type C nodes have some similarities with the associated data structures for Type B nodes, however, it is a much harder case since we cannot use the ordering on the $y$-coordinates of the right endpoints like for the Type B nodes. Instead we will precompute an approximation of $|F(v) \cap S|$ for every internal node $v$ in $\mathcal{T}$ and every square $S \in \mathcal{S}$ that lies entirely within $Int(v)$. Recall from Section~\ref{ssec:linear-number} that $\mathcal S$ is a set of size $O(n/\varepsilon^2)$ that is guaranteed to have a square that is a $(6+\varepsilon/8)$-approximation.

Let $\mathcal{S}(v)$ denote the subset of squares in $\mathcal{S}$ that lie entirely within $Int(v)$. The stored value for a square $S \in \mathcal{S}(v)$ is an upper bound on $|F(v) \cap S|$ and a lower bound on $|F(v)\cap S^+|$, where $S^+$ is the same as the one defined for Type B nodes. See Figure~\ref{fig:Type B}(c).
If $S^+$ intersects the left or right boundary of $Int(v)$ then perform the query as a Type B node instead of a Type C node. 

The associated data structure will be built on the set $F(v)$, hence, all segments will span the interval $Int(v)$. Let $s_1, \ldots , s_m$ be the segments in $F(v)$ and let $\mu(s_i)$ be the angle of inclination of $s_i$. Partition $F(v)$ into $\kappa_1$ sets $F_1(v), \ldots, F_{\kappa_1}(v)$ such that for any segment $s_i \in F_j(v)$ it holds that $(j-1)\cdot \frac{\pi}{4\kappa_1} \leq \mu(s_i) < j\cdot \frac{\pi}{4\kappa_1}$, with $\kappa_1=16\sqrt{2}/\varepsilon$. 


Using a combination of the approach we used for Type B nodes and a result by Agarwal~\cite{PartitionLineArrangements} we can prove the following lemma.

\begin{lemma} \label{lemma:type-C-sub}
Given a set $F_j(v)=\{s^j_1, \ldots , s^j_{n_1}\}$ of line segments (as defined above) and a set $\mathcal{S}(v)=\{S_1, \ldots , S_{n_2}\}$ of squares lying entirely within $Int(v)$, one can compute, in $O((n_1+n_2/\varepsilon)^{4/3} \polylog (n_1+n_2/\varepsilon))$ time using $O((n_1+n_2/\varepsilon)^{4/3}/ \log^{(2w+1)/3} (n_1+n_2/\varepsilon))$ space, where $w$ is a constant $<3.33$, a set of $n_2$ real values $\{M^j_1(v), \ldots , M^j_{n_2}(v)\}$ such that for every $i$, $1\leq i \leq n_2$ the following holds:
  $$|F_j(v) \cap S_i| \leq M^j_i(v) \leq |F_j(v) \cap S_i^+|.$$
\end{lemma}
\begin{proof}
For each square $S\in \mathcal{S}(v)$ partition the left side and the bottom side of $S$ into $2\kappa_2$ subsegments of equal length, where $\kappa_2 = 16/\varepsilon$. Note that a line segment in $F_j(v)$ can at most intersect one of the subsegments since the segments have positive slopes. 

Agarwal~\cite{PartitionLineArrangements} showed that given a set of $n_r$ red line segments and another set of $n_b$ blue line segments, one can count, for each red segment, the number of blue segments intersecting it in overall time $O(n^{4/3} \log^{(w+2)/3} n)$ using $O(n^{4/3}/ \log^{(2w+1)/3} n)$ space, where $n=n_r+n_b$ and $w$ is a constant $<3.33$. Note that there are faster algorithms, e.g.~\cite{c-chdc-93}, but to the best of the authors knowledge they do not return the number of red-blue intersection for each red segment.

Let the subsegments along the lower and left side of every square $S \in \mathcal{S}(v)$ be our red set of segments, hence $n_r=2\kappa_2 \cdot n_2$. Let the $n_b=n_1$ segments in $F_j(v)$ be our blue segments. Applying the algorithm by Agarwal~\cite{PartitionLineArrangements} immediately gives that in $O((n_1+\kappa_2\cdot n_2) \polylog (n_1+\kappa_2\cdot n_2))$ time we can compute the number of segments in $F_j(v)$ that intersect each subsegment of the squares in $\mathcal{S}(v)$.

For a subsegment $\ell$ of $S_i \in \mathcal{S}(v)$, let $a_{F_j}(\ell)$ be the number of segments in $F_j(v)$ that intersects $\ell$ and let $b_{F_J}(\ell)$ be the maximum length intersection between a possible segment in $F_j(v)$ and $S_i$. Set $M_i^j(\ell)=a_{F_j}(\ell) \cdot b_{F_j}(\ell)$, which is an upper bound on the total length of intersection between the segments in $F_j$ intersecting $\ell$ of $S_i$. To get an upper bound on the total length of intersection between the segments in $F_j(v)$ and $S_i$, denoted $M^j_i(v)$, we simply sum the $M^j_i(\ell)$-values for all subsegments $\ell$ of $S_i$. This is performed for each square $S_i \in \mathcal{S}(v)$, $1\leq i \leq n_2$.

Using a similar analysis as in Lemma~\ref{lemma:geometry} we get
$$|s \cap S_i| \leq b_{F_j}(\ell)  \leq |s \cap S_i^+|,$$  
and
$$|F_j(v) \cap S_i| \leq M^j_i(v) \leq |F_j(v) \cap S_i^+|,$$
which proves the lemma.
\end{proof}

For each set $F_j(v)$ apply Lemma~\ref{lemma:type-C-sub}, and for each square $S_i \in S(v)$ precompute the value $M_i = \sum_{j=1}^{\kappa_1} M_i^j(v)$. Store all the squares lying within the interval $Int(v)$ in a balanced binary search tree along with their precomputed values $M_i$. 
Note that each square of $\mathcal S$ can appear at most once on each level of the primary segment tree structure $\mathcal{T}$.
Furthermore, an edge $s \in \pi$ can straddle at most two intervals on a level, as a result we get that the total amount of time spent on one level of the segment tree to build the Type C associated data structures is $O((n/\varepsilon^3)^{4/3} \log^{(w+2)/3} (n/\varepsilon))$.
Since the number of levels in $\mathcal{T}$ is $O(\log n)$ the total time required to build all the Type C associated data structures is $O((n/\varepsilon^3)^{4/3} \log^{(w+5)/3} (n/\varepsilon))$. 
Putting all the pieces together we get:

\begin{lemma}
The Type C associated data structures for $\mathcal{T}$ can be computed in time 
$O((n/\varepsilon^3)^{4/3} \polylog (n/\varepsilon))$ 
using 
$O(n^{4/3}/\varepsilon^4)$
space. Given a query square $S\in \mathcal{S}(v)$, a value $M(v)$ can be returned in $O(\log (n/\varepsilon^2))$ time such that:
 $$|F(v) \cap S| \leq M(v) \leq |F(v) \cap S^+|.$$
\end{lemma}

\subsubsection{Putting it together}
In the previous section we showed how to construct a two-level data structure that uses a modified segment tree as the primary tree, and a set of associated data structures for all the internal nodes in the primary tree. The primary tree requires $O(n \log n)$ space, and the complexity of the associated data structures is dominated by the Type C nodes, that require $O((n/\varepsilon^3)^{4/3} \log^{(w+5)/3} (n/\varepsilon))$ time and $O((n/\varepsilon^3)^{4/3} / \log^{(2w+1)/3} (n/\varepsilon))$ space to construct. Given a query square $S\in \mathcal{S}$ a value $M$ is returned in $O(\frac{1}{\varepsilon^2}\cdot\log^2 n)$ time such that $\Upsilon(S) \leq M \leq \Upsilon(S^+).$

According to Lemma~\ref{lemma:linear-number-of-squares} there exists a square $S\in \mathcal{S}$ that has a packedness value that is within a factor of $(6+\varepsilon/8)$ smaller than the maximum packedness value of $\pi$. Using the data structure described in this section we get a $(6+\varepsilon/8)(1+\varepsilon/8)$-approximation. Since $\varepsilon$ is assumed to be at most 1, we finally get Theorem~\ref{theorem:complexAlgorithmThm}.
\newpage

\section{Concluding remarks}
In this paper we gave two approximation algorithm for the packedness value of a polygonal curve. The obvious question is if one can get a fast and practical $(1+\varepsilon)$-approximation. 

We also computed approximate packedness values for 16 real-world data sets, and the experiments indicate that the notion of $c$-packedness is a useful realistic input model for curves and trajectories. 

\bibliography{c_packed_ISAAC}

\begin{thebibliography}{10}

\bibitem{PartitionLineArrangements}
Pankaj~K. Agarwal.
\newblock Partitioning arrangements of lines {II:} applications.
\newblock {\em Discrete {\&} Computational Geometry}, 5:533--573, 1990.
\newblock \href {http://dx.doi.org/10.1007/BF02187809}
  {\path{doi:10.1007/BF02187809}}.

\bibitem{akklmss-cdsrp-08}
Pankaj~K. Agarwal, Rolf Klein, Christian Knauer, Stefan Langerman, Pat Morin,
  Micha Sharir, and Michael Soss.
\newblock Computing the detour and spanning ratio of paths, trees, and cycles
  in {2D} and {3D}.
\newblock {\em Discrete {\&} Computational Geometry}, 39(1):17--37, 2008.

\bibitem{F}
H.~Alt and M.~Godau.
\newblock Computing the {F}r{\'{e}}chet distance between two polygonal curves.
\newblock {\em International Journal of Computational Geometry}, 5:75--91,
  1995.

\bibitem{akw-cdmp-04}
Helmut Alt, Christian Knauer, and Carola Wenk.
\newblock Comparison of distance measures for planar curves.
\newblock {\em Algorithmica}, 38(1):45--58, 2004.

\bibitem{ahkww-fdcr-06}
Boris Aronov, Sariel Har-Peled, Christian Knauer, Yusu Wang, and Carola Wenk.
\newblock Fr{\'e}chet distance for curves, revisited.
\newblock In {\em Proceedings of the European Symposium on Algorithms ({ESA})},
  pages 52--63, 2006.

\bibitem{b-wwdtt-14}
Karl Bringmann.
\newblock Why walking the dog takes time: {F}rechet distance has no strongly
  subquadratic algorithms unless {SETH} fails.
\newblock In {\em 55th {IEEE} Annual Symposium on Foundations of Computer
  Science {(FOCS)}}, pages 661--670, 2014.

\bibitem{bk-iafdc-17}
Karl Bringmann and Marvin K{\"{u}}nnemann.
\newblock Improved approximation for {F}r{\'{e}}chet distance on $c$-packed
  curves matching conditional lower bounds.
\newblock {\em International Journal on Computational Geometry and
  Applications}, 27(1-2):85--120, 2017.

\bibitem{ck-dmpsa-95}
P.~B. Callahan and S.~R. Kosaraju.
\newblock A decomposition of multidimensional point sets with applications to
  $k$-nearest-neighbors and $n$-body potential fields.
\newblock {\em Journal of the ACM}, 42(1):67–90, 1995.

\bibitem{c-rcsi-86}
Bernard Chazelle.
\newblock Reporting and counting segment intersections.
\newblock {\em Journal of Computer and System Sciences}, 32(2):156--182, 1986.

\bibitem{c-chdc-93}
Bernard Chazelle.
\newblock Cutting hyperplanes for divide-and-conquer.
\newblock {\em Journal of Computer and System Sciences}, 9:145--158, 1993.

\bibitem{cdgnw-ammrf-11}
Daniel Chen, Anne Driemel, Leonidas~J. Guibas, Andy Nguyen, and Carola Wenk.
\newblock Approximate map matching with respect to the {F}r{\'{e}}chet
  distance.
\newblock In {\em Proceedings of the 13th Workshop on Algorithm Engineering and
  Experiments {(ALENEX)}}, pages 75--83, 2011.

\bibitem{kruger09}
P.~C. Cross, D.~M. Heisey, J.~A. Bowers, C.~T. Hay, J.~Wolhuter, P.~Buss,
  M.~Hofmeyr, A.~L. Michel, R.~G. Bengis, T.~L.~F. Bird, , et~al.
\newblock Disease, predation and demography: assessing the impacts of bovine
  tuberculosis on african buffalo by monitoring at individual and population
  levels.
\newblock {\em Journal of Applied Ecology}, 46(2):467--475, 2009.

\bibitem{msw-aotp-96}
R.~Silverman D.~Mount and A.~Wu.
\newblock On the area of overlap of translated polygons.
\newblock {\em Computer Vision and Image Understanding}, 64(1):53--61, 1996.

\bibitem{bkos-cgta-08}
M.~de~Berg, O.~Cheong, M.~J. van Kreveld, and M.~H. Overmars.
\newblock {\em Computational geometry: algorithms and applications, 3rd
  Edition}.
\newblock Springer, 2008.

\bibitem{b-lsbsp-00}
Mark de~Berg.
\newblock Linear size binary space partitions for uncluttered scenes.
\newblock {\em Algorithmica}, 28(3):353--366, 2000.

\bibitem{bsvk-rimga-02}
Mark de~Berg, Frank van~der Stappen, Jules Vleugels, and Matya Katz.
\newblock Realistic input models for geometric algorithms.
\newblock {\em Algorithmica}, 34(1):81--97, 2002.

\bibitem{dh-jydp-13}
Anne Driemel and Sariel Har{-}Peled.
\newblock Jaywalking your dog: Computing the {F}r{\'{e}}chet distance with
  shortcuts.
\newblock {\em {SIAM} Journal on Computing}, 42(5):1830--1866, 2013.

\bibitem{DBLP:journals/dcg/DriemelHW12}
Anne Driemel, Sariel Har{-}Peled, and Carola Wenk.
\newblock Approximating the {F}r{\'{e}}chet distance for realistic curves in
  near linear time.
\newblock {\em Discrete {\&} Computational Geometry}, 48(1):94--127, 2012.

\bibitem{dk-pefd-18}
Anne Driemel and Amer Krivosija.
\newblock Probabilistic embeddings of the {F}r{\'{e}}chet distance.
\newblock In {\em Proceedings of the 16th International Workshop Approximation
  and Online Algorithms}, pages 218--237, 2018.

\bibitem{DBLP:conf/cccg/Erickson95}
Jeff Erickson.
\newblock On the relative complexities of some geometric problems.
\newblock In {\em Proceedings of the 7th Canadian Conference on Computational
  Geometry}, pages 85--90. Carleton University, Ottawa, Canada, 1995.
\newblock URL: \url{http://www.cccg.ca/proceedings/1995/cccg1995\_0014.pdf}.

\bibitem{truckbus05}
Elias Frentzos, Kostas Gratsias, Nikos Pelekis, and Yannis Theodoridis.
\newblock Nearest neighbor search on moving object trajectories.
\newblock In {\em SSTD}, pages 328--345. Springer, 2005.

\bibitem{frechet}
M.~Maurice Fréchet.
\newblock Sur quelques points du calcul fonctionnel.
\newblock {\em Rendiconti del Circolo Matematico di Palermo}, 22:1--72, 1906.
\newblock \href {http://dx.doi.org/10.1007/BF03018603}
  {\path{doi:10.1007/BF03018603}}.

\bibitem{pigeon16}
Anna Gagliardo, Enrica Pollonara, and Martin Wikelski.
\newblock Pigeon navigation: exposure to environmental odours prior release is
  sufficient for homeward orientation, but not for homing.
\newblock {\em Journal of Experimental Biology}, pages jeb--140889, 2016.

\bibitem{bats15}
Luca Giuggioli, Thomas~J McKetterick, and Marc Holderied.
\newblock Delayed response and biosonar perception explain movement
  coordination in trawling bats.
\newblock {\em PLoS computational biology}, 11(3):e1004089, 2015.

\bibitem{gks-ahct-13}
J.~Gudmundsson, M.~J. van Kreveld, and F.~Staals.
\newblock Algorithms for hotspot computation on trajectory data.
\newblock In {\em 21st {SIGSPATIAL} International Conference on Advances in
  Geographic Information Systems}, pages 134--143. {ACM}, 2013.

\bibitem{ghps-pissf-20}
Joachim Gudmundsson, Michael Horton, John Pfeifer, and Martin Seybold.
\newblock A practical index structure supporting {F}réchet proximity queries
  among trajectories.
\newblock arXiv:2005.13773.

\bibitem{gs-faafm-15}
Joachim Gudmundsson and Michiel H.~M. Smid.
\newblock Fast algorithms for approximate {F}r{\'{e}}chet matching queries in
  geometric trees.
\newblock {\em Computational Geometry}, 48(6):479--494, 2015.

\bibitem{hr-fdre-14}
Sariel Har-Peled and Benjamin Raichel.
\newblock The {F}r\'{e}chet distance revisited and extended.
\newblock {\em ACM Transactions on Algorithms}, 10(1), 2014.

\bibitem{bbmw-fswd-17}
W.~Meulemans K.~Buchin, M.~Buchin and W.~Mulzer.
\newblock Four soviets walk the dog -- with an application to alt’s
  conjecture.
\newblock {\em Discrete \& Computational Geometry}, 58(1):180--216, 2017.

\bibitem{cats16}
Roland Kays, James Flowers, and Suzanne Kennedy-Stoskopf.
\newblock Cat tracker project.
\newblock \url{http://www.movebank.org/}, 2016.

\bibitem{vessel05}
Huanhuan Li, Jingxian Liu, Ryan~Wen Liu, Naixue Xiong, Kefeng Wu, and Tai-hoon
  Kim.
\newblock A dimensionality reduction-based multi-step clustering method for
  robust vessel trajectory analysis.
\newblock {\em Sensors}, 17(8):1792, 2017.

\bibitem{geo2012}
Microsoft.
\newblock Microsoft research asia, {GeoLife GPS} trajectories.
\newblock \url{http://www.microsoft.com/en-us/download/details.aspx?id=52367},
  2012.

\bibitem{mms-qsrs-97}
Joseph S.~B. Mitchell, David~M. Mount, and Subhash Suri.
\newblock Query-sensitive ray shooting.
\newblock {\em International Journal on Computational Geometry and
  Applications}, 7(4):317--347, 1997.

\bibitem{hurdat217}
NOAA.
\newblock National hurricane center, national oceanic and atmospheric
  administration, {HURDAT}2 atlantic hurricane database.
\newblock \url{http://www.nhc.noaa.gov/data/}, 2017.

\bibitem{os-rsplaf-96}
Mark~H. Overmars and A.~Frank van~der Stappen.
\newblock Range searching and point location among fat objects.
\newblock {\em Journal of Algorithms}, 21(3):629--656, 1996.

\bibitem{masked17}
Caroline~L Poli, Autumn-Lynn Harrison, Adriana Vallarino, Patrick~D Gerard, and
  Patrick~GR Jodice.
\newblock Dynamic oceanography determines fine scale foraging behavior of
  masked boobies in the gulf of mexico.
\newblock {\em PloS one}, 12(6):e0178318, 2017.

\bibitem{NBA16}
Rajiv Shah and Rob Romijnders.
\newblock Applying deep learning to basketball trajectories.
\newblock {\em arXiv preprint arXiv:1608.03793}, 2016.

\bibitem{soccer15}
STATS.
\newblock {STATS} {LLC} - data science.
\newblock \url{http://www.stats.com/data-science/}, 2015.

\bibitem{so-mpafo-94}
Frank van~der Stappen and Mark~H. Overmars.
\newblock Motion planning amidst fat obstacles (extended abstract).
\newblock In {\em Proceedings of the 10th Annual Symposium on Computational
  Geometry}, pages 31--40, 1994.

\bibitem{sobv-mpelo-98}
Frank van~der Stappen, Mark~H. Overmars, Mark de~Berg, and Jules Vleugels.
\newblock Motion planning in environments with low obstacle density.
\newblock {\em Discrete {\&} Computational Geometry}, 20(4):561--587, 1998.

\bibitem{DBLP:journals/talg/Vigneron14}
Antoine Vigneron.
\newblock Geometric optimization and sums of algebraic functions.
\newblock {\em {ACM} Trans. Algorithms}, 10(1):4:1--4:20, 2014.
\newblock URL: \url{https://doi.org/10.1145/2532647}, \href
  {http://dx.doi.org/10.1145/2532647} {\path{doi:10.1145/2532647}}.

\bibitem{gulls15}
Martin Wikelski, Elena Arriero, Anna Gagliardo, Richard~A Holland, Markku~J
  Huttunen, Risto Juvaste, Inge Mueller, Grigori Tertitski, Kasper Thorup,
  Martin Wild, et~al.
\newblock True navigation in migrating gulls requires intact olfactory nerves.
\newblock {\em Scientific reports}, 5:17061, 2015.

\bibitem{pentip06}
Ben~H Williams, Marc Toussaint, and Amos~J Storkey.
\newblock Extracting motion primitives from natural handwriting data.
\newblock In {\em ICANN}, pages 634--643. Springer, 2006.

\bibitem{taxiA11}
Jing Yuan, Yu~Zheng, Xing Xie, and Guangzhong Sun.
\newblock Driving with knowledge from the physical world.
\newblock In {\em Proc. of the 17th ACM SIGKDD Conf.}, pages 316--324. ACM,
  2011.

\bibitem{taxiB10}
Jing Yuan, Yu~Zheng, Chengyang Zhang, Wenlei Xie, Xing Xie, Guangzhong Sun, and
  Yan Huang.
\newblock T-drive: driving directions based on taxi trajectories.
\newblock In {\em Proceedings of the 18th ACM SIGSPATIAL Conference}, pages
  99--108. ACM, 2010.

\end{thebibliography}

\end{document}